\documentclass[pra, twocolumn,  superscriptaddress, groupedaddress]{revtex4}

\usepackage{color}
\usepackage{amsthm, amssymb, amsmath, amsfonts}
\usepackage{graphicx}
\usepackage{dcolumn}
\usepackage{bm}
\usepackage{subfigure}
\usepackage{verbatim}
\usepackage{amscd}
\usepackage{setspace}
\usepackage[]{graphicx}
\usepackage{enumerate}
\usepackage{enumitem}
\usepackage{thmtools,thm-restate}
\usepackage{bbold}
\usepackage{cleveref}

\def\complex{\mathbb{C}}

\def\natural{\mathbb{N}}
\def\integer{\mathbb{Z}}

\def\to{\rightarrow}

\newcommand{\microspace}{\mspace{0.5mu}}
\def \lket {\left|}
\def \rket {\right\rangle}
\def \lbra {\left\langle}
\def \rbra {\right|}
\newcommand{\ket}[1]{\lket\microspace #1 \microspace\rket}
\newcommand{\bra}[1]{\lbra\microspace #1 \microspace\rbra}

\newcommand{\kb}[1]{\ket{#1}\bra{#1}}

\newcommand{\1}{\mathbb{1}}

\newcommand{\cliff}[1]{\C^{( #1 )}}
\newcommand\gen[1]{\langle #1 \rangle}
\newcommand{\phase}[2]{\operatorname{P}_{#1}(#2)}
\newcommand{\wt}[1]{\operatorname{wt}\left( #1 \right)}

\def\C{\mathcal{C}}
\def\D{\mathcal{D}}

\def\H{\mathcal{H}}

\def\P{\mathcal{P}}

\def\S{\mathcal{S}}

\def\embeds{\hookrightarrow}
\def\conj{^{\dagger}}
\def\a{\mathbf{a}}
\def\b{\mathbf{b}}
\def\c{\mathbf{c}}
\def\e{\mathbf{e}}
\def\j{\mathbf{j}}

\def\v{\mathbf{v}}
\def\0{\mathbf{0}}

\newtheorem{corollary}{Corollary}
\newtheorem{definition}{Definition}
\newtheorem{lemma}{Lemma}
\newtheorem{theorem}{Theorem}
\newtheorem*{theorem*}{Theorem}
\newtheorem*{lemma*}{Lemma}

\begin{document}
\title{Diagonal gates in the Clifford hierarchy}

\author{Shawn X. \surname{Cui}}
\affiliation{University of California Santa Barbara, USA}
\affiliation{Stanford University, Palo Alto, USA}

\author{Daniel \surname{Gottesman}}
\affiliation{Perimeter Institute for Theoretical Physics, Waterloo, Canada}
\affiliation{CIFAR QIS Program, Toronto, Canada}

\author{Anirudh \surname{Krishna}}
\affiliation{Perimeter Institute for Theoretical Physics, Waterloo, Canada}
\affiliation{JARA IQI, RWTH Aachen University, Aachen, Germany}

\begin{abstract}
The Clifford hierarchy is a set of gates that appears in the theory of fault-tolerant quantum computation, but its precise structure remains elusive.
We give a complete characterization of the diagonal gates in the Clifford hierarchy for prime-dimensional qudits.  They turn out to be $p^m$-th roots of unity raised to polynomial functions of the basis state to which they are applied, and we determine which level of the Clifford hierarchy a given gate sits in based on $m$ and the degree of the polynomial.
\end{abstract}

\maketitle
\section{Introduction}

We expect that to build a large quantum computer, some sort of fault-tolerant
encoding will be necessary in order to deal with imperfections in quantum memories
and quantum gates.  Arguably the central result in the theory of fault-tolerant
quantum computation, the threshold theorem guarantees that it is possible to
construct reliable fault-tolerant quantum circuits provided the errors in state
preparation, gates, and measurements are below a certain threshold error rate.

The central idea behind this theorem is to encode quantum information into quantum
error correcting codes, the most common being stabilizer codes. To process
information, we can choose to use a transversal gate architecture and
measure Pauli observables. These gates prevent errors on a physical qubit
from spreading to others within an encoded block. Unfortunately, these gates alone
are insufficient to achieve universal quantum computation \cite{EastinKnill}.

Magic state injection is one common approach to overcome this limitation. Gottesman
and Chuang \cite{GottesmanChuang} explored what gates could be implemented via
teleportation-based state injection. They showed that there existed a class of gates
called the Clifford hierarchy that is intimately connected to fault tolerance and
state injection. The connection between state injection and the third level of the
Clifford hierarchy has been subsequently explored in \cite{Howard12, Bengtsson14,
Howard15}.  The Clifford hierarchy is also important in understanding the possible
transversal gates on stabilizer codes \cite{BravyiKonig,AndersonJochymOConnor}.
Although previous attempts have been made in \cite{Zeng08}, the full structure of
gates within the Clifford hierarchy is still not known.

In this paper, we make partial progress towards answering this question by giving
a complete characterization of the \emph{diagonal} gates in every level of the
Clifford hierarchy.  We focus on prime-dimensional qudits, but the result also
applies to qudits of prime-power dimension $p^r$ with a standard choice of Pauli
group, since their Clifford group and Clifford hierarchies are isomorphic to those of
$r$ $p$-dimensional qudits.  In particular, we show that if $U$ is a diagonal gate in
any level of the Clifford hierarchy for qudits of dimension $p$, it can be written as
\begin{equation}
  U = \sum_{j \in \integer_p} \exp \left(2\pi i \sum_m \delta_m(j)/p^m \right) \kb{j},
\end{equation}
where $\delta_m(j)$ is a polynomial over $\integer_{p^m}$ (a multivariate polynomial in
the case of multiple qudits).  The level of the Clifford hierarchy in which it
appears is determined by the largest value of $m$ that appears in the sum and the
degree of $\delta_m(j)$ for that $m$.

Section~\ref{sec:background} reviews some background material and establishes
terminology.  In section \ref{sec:singleQudit}, we prove the theorem for a single
qudit. We generalize this result to $n$ qudits in  section \ref{sec:nQudits} and make
some final comments in section~\ref{sec:conclusion}.

\section{Background}
\label{sec:background}
A single qudit of prime dimension $p$ is associated with the complex Euclidean space
$\complex^{p}$. Let $\omega = \exp\left( 2\pi i/p \right)$ denote the $p$-th root of
unity. The matrices $X$ and $Z$ are defined by their action on $\complex^{p}$: for $j
\in \integer_{p}$,
\begin{align}
  X \ket{j} = \ket{j + 1},\qquad Z\ket{j} = \omega \ket{j}~,
\end{align}
where the addition is performed with respect to the field $\integer_{p}$.

We will be dealing in this paper not just with powers of $\omega$, but with powers of
$\exp(2\pi i/p^m)$.


Let $\P$ denote the single qudit Pauli group
\begin{align}
  \P =
  \begin{cases}
    \gen{i\1, X, Z} &\mbox{if }p=2\\
    \gen{\omega \1, X, Z} &\mbox{if }p>2~.
  \end{cases}
\end{align}
We associate with $n$ qudits the Hilbert space $\H = \left(\complex^{p}\right)^{\otimes n}$. $\P_n := \P^{\otimes n}$
refers to the $n$-qudit Pauli group.  The Pauli group defines the first level in the Clifford hierarchy: $\cliff{1} = \{e^{i\phi}\} \cdot \P_n$.  We have added all global phases for later convenience.  We define
\begin{equation}
X(\v) = \bigotimes_{i=1}^n X^{v_i}
\end{equation}
and similarly for $Z(\v)$.  Here, $\v$ is an element of
$\integer_p^n$, an $n$-dimensional vector over $\integer_p$.

The group of automorphisms of the Pauli group is called the Clifford group and is
denoted $\cliff{2}$. These gates play a central role in the theory of quantum error
correction and fault tolerance. However, circuits composed entirely of gates
from $\cliff{2}$ are not universal for quantum computation.

To get around this problem, we need gates from the third level of the Clifford
hierarchy, $\cliff{3}$, defined as
\begin{align}
  \cliff{3} := \{U | UPU\conj \in \cliff{2},\; \forall P \in \P_n\}~.
\end{align}
Any gate from this set can be used to construct a universal quantum circuit in
conjunction with the Clifford group.

This can be generalized to define $\cliff{k}$, the $k^{th}$ level of the Clifford
hierarchy on $\H$
\begin{align}
  \cliff{k} := \{U | UPU\conj \in \cliff{k-1},\; \forall P \in \P_n\}~.
\end{align}
This set of gates was first defined by Gottesman and Chuang \cite{GottesmanChuang}
who showed that such gates can be implemented exactly via teleportation.

For $k \geq 3$, the set of gates in the Clifford hierarchy no longer forms a group.
However, diagonal Clifford operators $\cliff{k}_d \subset \cliff{k}$ in the
$k^{th}$ level of the Clifford hierarchy do form a group.

\begin{theorem}[\cite{Zeng08}]
$\cliff{k}_d$ is a group.
\end{theorem}

\begin{proof}
The proof works by induction on $k$.  Since $\cliff{2}$ is a group, so is
$\cliff{2}_d$.  To prove the result for larger $k$, the main observation is that if
unitary $U$ is diagonal (regardless if it is in $\cliff{k}_d$ or not), then
\begin{equation}
U X(\v) U^\dagger = V(\v) X(\v),
\end{equation}
with $V (\v)$ also a diagonal unitary.  Since $U \in \cliff{k}_d$ commutes with
$Z(\v)$, we only need to consider conjugation of $X(\v)$.

Now consider $U_1, U_2 \in \cliff{k}_d$. Then we have that $V_1(\v)$ and
$V_2(\v)$ are in $\cliff{k-1}_d$, so
\begin{align}
(U_1 U_2) X(\v) (U_1 U_2)^\dagger &= U_1 V_2(\v) X(\v) U_1^\dagger \\
&= V_2(\v) V_1(\v) X(\v),
\end{align}
since diagonal unitaries commute.  By the inductive hypothesis, $\cliff{k-1}_d$ is
a group, so $V_2(\v) V_1(\v) \in \cliff{k-1}_d$ and $U_1 U_2 \in
\cliff{k}_d$.

In addition, $U^\dagger X(\v) U = V'(\v) X(\v)$ implies that
\begin{equation}
[V'(\v)]^\dagger X(\v) = U X(\v) U^\dagger =  V(\v) X(\v),
\end{equation}
so $V'(\v) = V(\v)^\dagger \in \cliff{k-1}_d$, again by the inductive
hypothesis.  This implies that $U^\dagger \in \cliff{k}_d$.
\end{proof}

\section{Single-qudit diagonal unitary gates and the Clifford
hierarchy}\label{sec:singleQudit}

Let $p$ be some prime number and $m \in \natural$ be a fixed natural number. The ring
$\integer_{p^{m}}$ is defined as
\begin{align}
  \integer_{p^{m}} := \{0, 1, \cdots, p^{m} - 1\}~.
\end{align}
Any element $c \in \integer_{p^m}$ can be expressed as
\begin{align}
  c_0 + c_1 p + \cdots + c_{m-1} p^{m-1}~,
\end{align}
where $\{c_i\}_{i=0}^{m-1}$ are some constants in $\integer_p$.

Let $\Theta: \integer_{p} \embeds \integer_{p^{m}}$ be an arbitrary function.
It can be constructed using polynomials of degree at most $p-1$. This can be seen as
follows. Let $\delta_{k}(j)$ be a delta function such that it is $1$ when $j
= k$ and $0$ otherwise. $\Theta$ can then be expressed as
\begin{align}
  \Theta(j) = \sum_{k} \theta_k \delta_{k}(j)~,
\end{align}
for some constants $\theta_k \in \integer_{p^{m}}$. $\delta_{k}(j)$ is a polynomial of
degree at most $p-1$ since it can be expressed as
\begin{align}
  \delta_{k}(j) = \prod_{\substack{k' \in \integer_p\\ k' \neq k}} \frac{(j - k')}{(k-k')}~.
\end{align}

We shall be interested in studying diagonal unitary operators of the form
\begin{align}
  U = \sum_{j \in \integer_{p}} \exp\left( \frac{2 \pi i}{p^{m}} \Theta(j) \right)
  \kb{j}~.
\end{align}
In this context, we shall refer to $m$ as the precision of the unitary $U$.
Note that all unitary operators $U$ of precision $m$ can be expressed in the manner
above.

We begin by focusing on unitaries constructed using monomial $\Theta$.
\begin{definition}
  For $m \in \natural$, $1 \leq a \leq p-1$, the diagonal unitary gate $U_{m,a}$ is defined as
     \begin{align}
       U_{m,a} := \sum_{j \in \integer_{p}} \exp\left( \frac{2\pi i}{p^{m}} j^{a} \right)
       \kb{j}~.
     \end{align}
\end{definition}
We ignore $a = 0$ because such unitaries are only a constant phase times the
identity operator.

We then define the set of diagonal unitaries $\D_{m,a}$ recursively:
\begin{definition}
\begin{align}
  \D_{m,a} = \gen{ U_{m,b} }_{b=1}^{a} \cdot \{e^{i\phi}\} \cdot \D_{m-1,p-1}~.
\end{align}
\end{definition}
As mentioned earlier, polynomials of degree $p-1$ can be used to construct arbitrary
functions $\Theta: \integer_p \embeds \integer_{p^{m}}$. Hence, $\D_{m,p-1}$ can be
used to construct any diagonal unitary of precision $m$.

Note that $\D_{1,1} = \gen{Z}$ is simply the set of all diagonal Pauli operators with global phase $\phi$.
Hence we may write
\begin{align}
  \D_{1,1} = \cliff{1}_d~.
\end{align}

Among all the diagonal unitary gates, we single out a special class of gates called
phase gates:
\begin{definition}
  For $m \in \natural$, $\phase{m}{k}$ is the phase gate that changes the phase of $\ket{k}$:
     \begin{align}
       \phase{m}{k} = \sum_{\substack{j=0\\ j \neq k}}^{p-1} \kb{j} + \exp\left(
       \frac{2\pi i}{p^{m}} \right) \kb{k}~.
     \end{align}
\end{definition}
Phase gates are not actually distinct diagonal unitary gates. Since the function
$\delta_{k}(j)$ can be represented as a polynomial of degree $p-1$, the phase gate
$\phase{m}{k} \in \D_{m,p-1}$. Nevertheless, it will be helpful to be able to refer to $\phase{m}{k}$ directly.

The main result of this section is the following theorem:
\begin{theorem}\label{thm:diagIsCliff}
  For $m \in \natural$, and $1 \leq a \leq p-1$,
  \begin{align}
    \D_{m,a} = \cliff{(p-1)(m-1) + a}_d ~.
  \end{align}
\end{theorem}

To prove this, we shall break the result into two lemmas, each showing containment of one
group in the other.
\begin{lemma}\label{lem:diagInCliff}
For $m \in \natural$, and $1 \leq a \leq p-1$,
  \begin{align}
    \D_{m,a} \subseteq \cliff{(p-1)(m-1) + a}_d ~.
  \end{align}
\end{lemma}
\begin{proof}
The proof proceeds via induction on both $m$ and $a$.

\textbf{Base case: }
By definition, $\D_{1,1}$ is the group of all diagonal Pauli operators and therefore
\begin{align}
  \D_{1,1} = \cliff{1}_d ~.
\end{align}
This implies the weaker result
\begin{align}
  \D_{1,1} \subseteq \cliff{1}_d ~.
\end{align}

\textbf{Induction on $a$:}

Suppose we have proved
\begin{enumerate}
  \item $\forall m' < m$, $\forall \, b \in \integer_{p}$, that
    \begin{align}
      \D_{m',b} \subseteq \cliff{(p-1)(m'-1) + b}_d ~.
    \end{align}
  \item and $\forall \, a'$ such that $1 \leq a' < a \leq p-1$, that
         \begin{align}
	   \D_{m,a'} \subseteq \cliff{(p-1)(m-1) + a'}_d~.
	 \end{align}
\end{enumerate}

Consider the conjugation
\begin{align}
  U_{m,a}X U_{m,a}\conj = \sum_{j=1}^{p-1}
  &\exp\left( \frac{2\pi i}{p^{m}}\left[j^{a} - (j-1)^{a}\right]\right)
  \ket{j}\bra{j-1} \nonumber \\
  + &\exp\left(-\frac{2\pi i}{p^{m}}(p-1)^{a} \right) \ket{0}\bra{p-1}\\
  = \sum_{j=1}^{p-1}
  &\exp\left( \frac{2\pi i}{p^{m}}\left(-\sum_{d=0}^{a-1}
  c_{d}j^{d}\right)\right)\ket{j}\bra{j-1} \nonumber\\
  + &\exp\left(-\frac{2\pi i}{p^{m}}\left( \sum_{d=0}^{a} c_{d}p^{d}
  \right)\right)\ket{0}\bra{p-1} \label{eq:sumOverJ}~,
\end{align}
where
\begin{align}
  c_d = \left(
   \begin{matrix}
     a\\
     d\\
   \end{matrix}
   \right)(-1)^{a-d} ~.
\end{align}
We have separated the sum over $j$ into two parts because this allows us to write it as
a product of gates that can be easily identified. First, note that the entire
expression contains a constant phase
\begin{align*}
  \exp\left( \frac{2\pi i}{p^{m}} (-1)^{a+1} \right)
\end{align*}
that arises from the $d=0$ terms and can be removed.

In equation (\ref{eq:sumOverJ}), the sum over $j$ arises from a diagonal unitary $W_{m,a-1}$ times
$X$, where $W_{m,a-1} \in \D_{m,a-1}$: this unitary has the form
\begin{align}
  W_{m,a-1} = \sum_{j \in \integer_p} \exp\left( \frac{2\pi i}{p^{m}} a j^{a-1}
  \right) \kb{j}.
\end{align}
We have ignored terms of the form
\begin{align}
  \frac{1}{p^{n}}j^{d}
\end{align}
if $n < m$ or if $n = m$ but $d < a-1$ since they are in $\D_{m,d} \subseteq \D_{m,a-1}$ and will therefore not affect the level of the hierarchy.

The next term of the expression (\ref{eq:sumOverJ}) is a product of phase gates $\phase{m-d}{0}$ times
$X$, where $d$ is at least $1$. To pin down which level of the Clifford hierarchy
$U_{m,a}$ lies in, we only need to consider the finest phase rotations i.e. the terms
with the largest precision; the rest of the gates are lower in the hierarchy and can
safely be ignored.

With this observation, we can write the above expression as
\begin{align}
 \exp\left( \frac{2\pi i}{p^m} (-1)^{a+1} \right) \phase{m-1}{0}W_{m,a-1} X~.
\end{align}
This can be further simplified. The phase gate $\phase{m-1}{0} \in \D_{m-1,p-1}
\subseteq \D_{m,a-1}$ and therefore the product $\phase{m-1}{0}W_{m,a-1} := V_{m,a-1}
\in \D_{m,a-1}$. Hence, the above expression is
\begin{align}\label{eq:finalStep}
  \exp\left( \frac{2\pi i}{p^m} (-1)^{a+1} \right)V_{m,a-1}X~.
\end{align}

Using the inductive hypothesis, we know that
\begin{align}
  &V_{m,a-1} \in \cliff{(p-1)(m-1) + (a-1)}_d \nonumber \\
  \implies &U_{m,a} \in \cliff{(p-1)(m-1) + a}_d~.
\end{align}
Therefore,
\begin{align}
  \D_{m,a} \subseteq \cliff{(p-1)(m-1) + a}_d~.
\end{align}

\textbf{Induction on $m$:}

Suppose we have shown that $\forall m' < m$ and $a \in \integer_{p}$,
\begin{align}
  \D_{m',a} \subseteq \cliff{(p-1)(m'-1) + a}_d ~.
\end{align}
Consider the conjugation
\begin{align}
  U_{m,1}XU_{m,1}\conj
  = \sum_{j=1}^{p-1}&\exp\left( \frac{2\pi i}{p^{m}}
  \right)\ket{j}\bra{j-1}\nonumber\\
  + &\exp\left(-\frac{2\pi i}{p^{m}} (p-1) \right)
  \ket{0}\bra{p-1}\\
  = &\exp\left( \frac{2\pi i}{p^m} \right) \phase{m-1}{0}^{-1}X ~.
\end{align}

Since the phase gate $\phase{m-1}{0} \in \D_{m-1,p-1}$, the inductive
hypothesis stipulates
\begin{align}
  \phase{m-1}{0} \in \cliff{(p-1)(m-1)}_d \implies U_{m,1} \in \cliff{(p-1)(m-1) + 1}_d ~.
\end{align}
Therefore,
\begin{align}
  \D_{m,1} \subseteq \cliff{(p-1)(m-1) + 1}_d~.
\end{align}
\end{proof}

\begin{lemma}\label{lem:cliffInDiag}
  $\D_{m,a} \supseteq \cliff{(p-1)(m-1) + a}_d $.
\end{lemma}
\begin{proof}

For $m, m' \geq 1, \,  1 \leq a,a' \leq p-1$, define $(m',a') < (m,a)$ if $m' < m$ or $m' = m$ and $a' < a$. Clearly this defines a total ordering on the set of pairs $\{(m,a)\}$, and $(m',a') < (m,a)$ if and only if $(m'-1)(p-1) + a' < (m-1)(p-1)+a$.

We shall prove this lemma by induction on $(m,a)$ relative to this ordering.

  \textbf{Base case $(m,a) = (1,1)$:} By definition $\D_{1,1} = \cliff{1}_d$ and therefore,
  \begin{align}
    \D_{1,1} \supseteq \cliff{1}_d ~.
  \end{align}

  \textbf{Induction on $(m,a)$:} Suppose we have shown
    $\forall \, (m',a') < (m,a)$ that
  \begin{align}
    \D_{m',a'} \supseteq \cliff{(p-1)(m'-1) + a'}_d~.
  \end{align}

  Suppose $U \in \cliff{(p-1)(m-1) + a}_d $.
  Let us express $U$ as
  \begin{align}
    U = \sum_{j \in \integer_{p}} \exp\left( 2\pi i \cdot \theta(j) \right) \kb{j}~.
  \end{align}
  Without loss of generality we can let $\theta(0) = 0$, absorbing the difference into a global phase. We would like to show that
  \begin{align}
    U \in \D_{m,a}~.
  \end{align}

  For some $\phi \in [0, 1)$, we are guaranteed the existence of a unitary $V \in
  \cliff{(p-1)(m-1) + (a-1)}_d$ such that
  \begin{align}\label{eq:constraint}
    U X U\conj = e^{2 \pi i\phi} V X~.
  \end{align}

  From the inductive hypothesis, $V$ is an element of $\D_{m,a-1}$ if $a \geq 2$, and an element of $\D_{m-1,p-1}$ if $a = 1$. Let $\Delta\theta$
  denote the function
  \begin{align}
    \Delta\theta(j) = \begin{cases}
      & \theta(j) - \theta(j-1) \mbox{ if } j \in \{1,\cdots,p-1\}\\
      & \theta(0) - \theta(p-1) \mbox{ if } j = 0~.
  \end{cases}
  \end{align}

  Together with Equation (\ref{eq:constraint}), this implies that for $j=0,1,\cdots, p-1,$ we must have, for some $\mu_{(m',a')} \in \integer_{p}$,
  \begin{align}\label{eq:recurse}
    \Delta\theta(j) &= \sum\limits_{(m',a')<(m,a)}\mu_{(m',a')}\frac{j^{a'}}{p^{m'}} + \phi \mod  1.
  \end{align}

  Let
\begin{align}
\sum(j,a') = \sum\limits_{k=1}^{j} k^{a'}.
\end{align}

Then adding up the $p$ equations in (\ref{eq:recurse}) we obtain
  \begin{align}\label{eq:sum_phi}
  \sum\limits_{(m',a')<(m,a)} \frac{\mu_{(m',a')}}{p^{m'}} \sum(p-1,a') + p \phi = 0 \mod 1.
  \end{align}

Since $\integer_p^{\times}$ is a cyclic group, it is direct to show that
\begin{equation}\label{eq:sum_p-1}
\sum(p-1,a') =
\begin{cases}
p-1 \mod p, & a' = p-1 \\
0 \mod p, & a' \neq p-1
\end{cases}
\end{equation}

Substituting (\ref{eq:sum_p-1}) into (\ref{eq:sum_phi}), we know  that there exist $\nu_{m',a'} \in \integer, \, w \in \integer$, such that

\begin{align}
\phi = \sum\limits_{\substack{(m',a')<(m,a) \\ a' \neq p-1}} \frac{\nu_{(m',a')}}{p^{m'}} + \sum\limits_{\substack{(m',a')<(m,a) \\ a' = p-1}}\frac{\nu_{(m',a')}}{p^{m'+1}}  + \frac{w}{p}.
\end{align}

Since $(m',p-1) < (m,a)$ implies $m'+ 1 \leq  m$, there exists $u \in \integer$ such that
\begin{align}\label{eq:phi}
\phi = \frac{u}{p^m}.
\end{align}

Next, $\theta(j)$ can be derived from the inductive formula in equation (\ref{eq:recurse}),
\begin{align}\label{eq:theta}
\theta(j) &= \sum\limits_{(m',a')<(m,a)}\frac{\mu_{(m',a')}}{p^{m'}} \sum(j,a') + j \phi \mod  1 \nonumber \\
          &= \sum\limits_{(m',a')<(m,a)}\frac{\mu_{(m',a')}}{p^{m'}} \sum(j,a') + \frac{uj}{p^m} \mod  1.
\end{align}

Faulhaber's formula~\cite{Faulhaber} on sums of powers of positive integers states that
\begin{align}
\sum(j,a') = \frac{1}{a'+1} \sum_{k=0}^{a'} (-1)^k \binom{a'+1}{k}B_k j^{a'+1-k},
\end{align}
where $B_k$'s are the Bernoulli numbers and $B_1 = -\frac{1}{2}$. We use the following two facts on Bernoulli numbers:
\begin{enumerate}
\item $B_{2n+1} = 0,\, n \geq 1$.
\item The denominator of $B_{2n}$ is the product of all prime numbers $q$ such that $q-1$ divides $2n$.
\end{enumerate}

In the following we discuss some properties of $\sum(j,a')$ in two cases.

\textbf{Case 1: $a' \neq p-1$}.

Since $a' \leq p-2$, $p$ can not be a divisor of the denominator of any $B_{2n}$ for $2n \leq a'$. Let $L$ be the least common multiplier of the denominators of $\{B_{2n}, 2n \leq a'\} \cup \{B_1\}$.  Then $L$ is coprime to $p$, and we have
\begin{align}
L \cdot (a'+1) \sum(j,a') = \sum\limits_{k=0}^{a'} a_k j^{a'+1-k},\, a_k \in \integer~.
\end{align}

Let $I_{m'} \in \integer$ be the inverse of $L \cdot (a'+1)$ modulo $p^{m'}$, then
\begin{align}\label{eq:sumj1}
\sum(j,a') = \sum\limits_{k=0}^{a'} I_{m'}a_k j^{a'+1-k} \mod p^{m'}.
\end{align}

\textbf{Case 2: $a' = p-1$}.

In this case, $\sum(j,a')$, just like any function from $\integer_p$ to $\integer_{p^{m'}}$, can be written as a polynomial $\Theta_{a'}(j)$ of degree at most $p-1$ over $\integer_{p^{m'}}$.

Finally, combining equations (\ref{eq:theta}) and (\ref{eq:sumj1}), we have

\begin{align}\label{eq:phi_final}
\theta(j) &= \sum\limits_{\substack{(m',a')<(m,a) \\ a' \neq p-1}}\frac{\mu_{(m',a')}}{p^{m'}} \sum\limits_{k=0}^{a'} I_{m'}a_k j^{a'+1-k} \nonumber\\ 
&+ \sum\limits_{\substack{(m',a')<(m,a) \\ a' = p-1}}\mu_{(m',a')} \frac{\Theta_{a'}(j)}{p^{m'}} \nonumber \\ &+ \frac{uj}{p^m} \mod  1~.
\end{align}

Again using the fact that $(m',p-1) < (m,a)$ implies $m'+ 1 \leq  m$, we know that
the terms in the second line of the above equation sit in $\D_{m',p-1} \subset
\D_{m,a}$. It is easy to see the other terms in the equation are also in $\D_{m,a}$. Thus $U \in \D_{m,a}$.

\end{proof}

Since $\cliff{k}_d$ is an Abelian group, it can be written as a product of cyclic groups.  Now that we know its structure, it is straightforward to determine this decomposition explicitly.

\begin{corollary}
For $a\leq p-1$,
\begin{equation}
\cliff{a}_d = \D_{1,a} \cong U(1) \times \integer_{p}^a~.
\end{equation}
For $m > 1$,
\begin{equation}
\cliff{(p-1)(m-1) + a}_d = \D_{m,a} \cong U(1) \times \integer_{p^m}^a \times \integer_{p^{m-1}}^{p-a-1}~.
\end{equation}
\end{corollary}

\begin{proof}
$\gen{U_{m,a}}$ is isomorphic to $\integer_{p^m}$.  It contains $\gen{U_{m',a}}$ for $m' < m$ but not $\gen{U_{m',a'}}$ for any $a' \neq a$.  Therefore, each degree of polynomial with prefactor $1/p^m$ corresponds to one factor of $\integer_{p^m}$, and each degree with prefactor $1/p^{m-1}$ corresponds to one factor of $\integer_{p^{m-1}}$.  Lower values of $m' < m-1$ do not give additional factors because all degrees of polynomials up to $p-1$ are already present for $m$ or $m-1$.  There is also a global phase, isomorphic to $U(1)$.
\end{proof}

\section{$n$ qudit diagonal gates and the Clifford hierarchy}\label{sec:nQudits}

In this section, we shall generalize the above results to $n$ qudits. $\a,\b,\cdots$
shall denote vectors in $\integer_{p}^{n}$. The weight of a vector $\a \in
\integer_p^n$, is defined as $\wt{\a} := \sum_{i=1}^{n} a_i$. A basis element of $\H
= \left( \complex^{p} \right)^{\otimes n}$ is represented as $\ket{\j}
= \bigotimes_{i=1}^{n} \ket{j_i}$. For $i \in \{1,\cdots,n\}$, let $\e_i \in
\integer_{p}^{n}$ be the vector whose $i$-th component is $1$ and the rest are $0$.

Let $\Theta: \integer_{p}^{n} \to  \integer_{p^{m}}$ be some function. Following
a similar line of reasoning as in the previous section, we can show that any such
function can be constructed using polynomials of degree at most $n(p-1)$.

Any diagonal unitary of precision $m$ on $n$ qudits can be expressed as
\begin{align}
  U = \sum_{\j}\exp\left( \frac{2\pi i}{p^{m}} \Theta(\j) \right) \kb{\j}~.
\end{align}

As before, we shall start with unitaries whose exponents only contain monomial terms.
\begin{definition}
  For $m \in \natural$ and $\a \in \integer_{p}^{n}$, such that $0 \leq a_i \leq p-1$,
      \begin{align}
      U_{m,\a} := \sum_{\j} \exp\left( \frac{2\pi i}{p^m}
      j_1^{a_1}\cdots j_n^{a_n} \right) \kb{\j}~.
      \end{align}
\end{definition}
Similar to the single qudit case, $j_1^{a_1}\dots j_n^{a_n}: \integer_p^n
\to \integer_{p^m}$, i.e. the range of these monomials is $\integer_{p^m}$.

The set $\D_{m,\a}$ shall denote the set of diagonal unitary operators whose
exponents are multivariate polynomials of degree $\a$.
\begin{definition}
  For $m \in \natural$, $\a \in \integer_{p}^{n}$, and vectors $\b \in
  \integer_{p}^{n}$ such that $b_i \leq a_i$ for all $i \in [n]$ and $\wt{\b}
  < \wt{\a}$, and for any vectors $\c \in \integer_{p}^{n}$ such that $1 \leq c_i \leq
  p-1$, and $\wt{\c} \leq \wt{\a} + (p-1)$
  \begin{align}
    \D_{m,\a} := \gen{U_{m,\a}} \cdot \{e^{i\phi}\} \cdot \prod_{\b} \D_{m,\b} \cdot \prod_{\c} \D_{m-1,\c}~.
  \end{align}
\end{definition}

Note that $D_{1,\e_i} = \gen{Z(\e_i)} \cdot \{e^{i\phi}\}$ is the set of diagonal Paulis on the $i$th qudit with a global phase.

\begin{definition}
  For $w \in \natural$, let $\S_w$ denote the set
  \begin{align}
    \S_w = \{(m,\a) | (p-1)(m-1) + \wt{\a} = w \}~.
  \end{align}
  We then define
  \begin{align}
    \D_{w} := \prod_{(m,\a) \in \S_w} \D_{m,\a}~.
  \end{align}
\end{definition}

The main result of this section is the following theorem
\begin{theorem}\label{thm:diagIsCliffnQudit}
  For $w \in \natural$,
  \begin{align}
    \D_{w} = \cliff{w}_d ~.
  \end{align}
\end{theorem}

As in the single qudit case, we shall break the proof of the theorem into two lemmas.
\begin{lemma}\label{lem:diagInnQuditCliff}
  For $w \in \natural$,
      \begin{align}
	\D_{w} \subseteq \cliff{w}_d ~.
      \end{align}
\end{lemma}
\begin{proof}
  \textbf{Base case: }
  By definition
  \begin{align}
    \prod_{i=1}^{n} \D_{1,e_i} = \cliff{1}_d~.
  \end{align}
  Therefore,
  \begin{align}
    \D_{1,e_i} \subseteq \cliff{1}_d~.
  \end{align}

  \textbf{Inductive step: }
  For $w' < w$, suppose we have shown that
  \begin{align}
    \D_{w'} \subseteq \cliff{w'}_d~.
  \end{align}
  Let $m \in \natural$ and $\a \in \integer_{p}^{n}$ such that
  \begin{align}
  (p-1)(m-1) + \wt{\a} = w~.
  \end{align}

There are three components making up $\D_{m,\a}$.  The first component is $\gen{U_{m,\a}}$, with $(m,\a)$ satisfying the above constraint.  The second component is elements of $\D_{m,\b}$ and by the condition on $\b$, $(p-1)(m-1) + \wt{\b} = w' < w$.  Thus, $\D_{m,\b} \subseteq \D_{w'}$ and the inductive hypothesis implies that $\D_{m,\b} \subseteq \cliff{w'}_d$.  The third component is $\D_{m-1,\c}$ and
\begin{align}
(p-1)[(m-1)-1] + \wt{\c} &\leq (p-1)(m-1) + \wt{\a} \nonumber \\
&= w.
\end{align}
Therefore, to show that $\D_{m,\a} \subseteq \cliff{w}_d$, it suffices to show that $U_{m,a} \in \cliff{w}$.
 To this end, consider
  \begin{align}
    U_{m,\a} X(\e_1) U_{m,\a}\conj
    = V X(\e_1)~.
  \end{align}

  \textbf{Case 1:}
  If $a_1 > 1$, then it is straightforward to show as in lemma~\ref{lem:diagInCliff} that $V \in \D_{m,\b}$ where $\b  = \a - \e_1$. The inductive hypothesis guarantees $V \in \cliff{w-1}_d$.

  \textbf{Case 2:}
  If $a_1 = 1$, then we can show that $V$ is a product of two gates, $V_L$ and
  $V_R$. $V_L \in \D_{m,\b}$ where $\b = \a - \e_1$ as before; $V_R \in \D_{m-1,\c}$
  where $\c = (p-1,b_2,\dots,b_n)$.  Since $\wt{\c} = \wt{\a} + (p-2)$, $(m-1,\c) \in \S_{w-1}$.
  The product of $V_L$ and $V_R$ always lies in $\D_{w-1}$ and hence, by the
  inductive hypothesis, $V_L \cdot V_R \in \cliff{w-1}_d$.

The same argument works for conjugation of $X(\e_i)$ for $i \neq 1$.  If
  \begin{equation}
    U_{m,\a} X(\e_i) U_{m,\a}\conj
    = V_i X(\e_i)~,
  \end{equation}
then
\begin{equation}
    U_{m,\a} X(\v) U_{m,\a}\conj
    = \prod_i V_i^{v_i} X(\v)~.
\end{equation}
Since $\cliff{w-1}_d$ is a group, $\prod_i V_i^{v_i} \in \cliff{w-1}_d$ as well.
  This implies that
  \begin{align}
    U_{m,\a} \in \cliff{w}_d~.
  \end{align}

\end{proof}

\begin{lemma}\label{lem:nQuditCliffInDiag}
  For $w \in \natural$,
  \begin{align}
    \D_{w} \supseteq \cliff{w}_d ~.
  \end{align}
\end{lemma}
\begin{proof}
  \textbf{Base case: }
  By definition, $\prod_i \D_{1,\e_i} = \cliff{1}_d$ and therefore,
  \begin{align}
    \D_{1} \supseteq \cliff{1}_{d}~.
  \end{align}

  Suppose we have shown that for $w \in \natural$, $w' < w$ that
  \begin{align}
    \D_{w'} \supseteq \cliff{w'}_d ~.
  \end{align}
  Let $U \in \cliff{w}_d$.  It can be expressed as
  \begin{align}
   U = \sum_{\j} \exp\left( 2\pi i \cdot \theta(\j) \right) \kb{\j}~,
  \end{align}
  for some function $\theta$.

  For some $\phi \in [0,2)$, there exists an operator $V \in
  \cliff{w-1}_d$ such that
  \begin{align}
    UX(\e_1) U\conj &= e^{2\pi i\phi}V X(\e_1)~.
  \end{align}
  We begin by considering only the conjugation with $X(\e_1)$ for simplicity.

  Let $\Delta_i \theta$ denote the function
  \begin{align*}
    &\Delta_i \theta(j_1,\cdots,j_i,\cdots,j_n) = \\
      &\theta(j_1,\cdots, j_i, \cdots, j_n) - \theta(j_1,\cdots,j_i -1, \cdots,j_n)\\
    &\Delta_i \theta(j_1,\cdots,j_i = 0, \cdots, j_n) = \\
    &\theta(j_1,\cdots, 0, \cdots,j_n) - \theta(j_1,\cdots,p -1, \cdots,j_n)~.
  \end{align*}

  From our inductive assumption it follows that $V \in \D_{w-1}$.
  Hence, there exists $N \in \natural$ such that
  $V$ can be expressed as the product of unitaries $\{V_x\}_{x=1}^{N} \in
  \D_{w-1}$ where each unitary can be expressed as
  \begin{align}
    V_x = \sum\limits_{\j}\exp\left( \frac{2\pi i}{p^{m_x}} j_1^{b_{x,1}}\dots j_n^{b_{x,n}} \right)\kb{j}~,
  \end{align}
with $(m_x, \b_x) \in \S_{\alpha}$, $\alpha \leq w-1$.
  That is,
  \begin{align}
    (p-1)(m_x - 1) + \wt{\b_x} = \alpha < w~.
  \end{align}

  We can then express the polynomial $\Delta_1 \theta$ as
  \begin{align}\label{eq:recursen}
    \Delta_1 \theta(\j) = \phi + \sum_x \frac{1}{p^{m_x}}
    \mu_{m_x,\b_x} \cdot j_1 ^{b_{x,1}} \dots j_n^{b_{x,n}} \mod 1
  \end{align}
  for some constants $\mu_{m_x,\b_x} \in \integer_{p}$. We have ignored terms of the form
  \begin{align}
    \frac{1}{p^{n}} j^{\c}~,
  \end{align}
  where $n < m_x$ or $n = m_x$ and $\wt{\c} < \wt{\b_x}$.

As in the single-qudit proof (lemma~\ref{lem:cliffInDiag}), we find $\phi = u/p^m$, where $m = \max{m_x}$.  We again apply Faulhaber's result.  The argument is the same as the single-qudit case, but this time, we find multiple leading order terms in $\theta$:
  \begin{align}
    \theta(\j) =
    \sum_{x} \frac{1}{p^{\tilde{m}_x}} \alpha_{\tilde{\b}_x} j_1^{\tilde{b}_{x,1}}
    \dots j_n^{\tilde{b}_{x,n}} + \frac{uj_1}{p^m}.
  \end{align}
  for some constants $\alpha_{\a} \in \integer_{p}$ and tuples
  $(\tilde{m}_x,\tilde{\b}_x)$ such that either $\tilde{m}_x = m_x$ and $\tilde{b}_1
  = b_1 + 1$ or $\tilde{m}_x = m_x + 1$ and $\tilde{b_{1}} = 1$. This means that
  these tuples obey
  \begin{align}
    (p-1)(\tilde{m}_x - 1) + \wt{\tilde{\b}_x} = \alpha + 1 \leq w~.
  \end{align}

The other difference from the single-qudit case is that there are ``constants'' that appear in the proof of lemma~\ref{lem:cliffInDiag} which in the multiple-qudit case are actually functions of $j_2$ through $j_n$, just not $j_1$.  For most of these functions, their value in $\theta$ is fixed by the corresponding polynomials in $V$, and therefore they are polynomials in $\theta$ as well.  However, $\theta(0)$ disappears completely in $\Delta \theta$ and now cannot be absorbed into the global phase either. 

By repeating the argument for $X(\e_j)$ for $j \in [n]$, we find that $\theta(0)$ and therefore $U$ can be expressed as the product of unitaries
  $U_{\tilde{m}^j_y,\tilde{\b}^j_y}$ such that
  \begin{align}
    (p-1)(\tilde{m}^j_y -1) + \wt{\tilde{\b}^j_y} \leq w~.
  \end{align}
  Therefore,
  \begin{align}
    U \in \D_{w}~,
  \end{align}
  which implies
  \begin{align}
    \D_{w} \supseteq \cliff{w}_d
  \end{align}
  as desired.
\end{proof}

Again, we can express $\cliff{w}_d$ as a product of cyclic groups.
\begin{corollary}
Let
\begin{equation}
m_{w,\a} = \left\lfloor \frac{w - \wt{\a}}{p-1} \right\rfloor~.
\end{equation}
Then
\begin{equation}
\cliff{w}_d \cong U(1) \times \prod_{\a | \wt{\a} \leq w} \integer_{p^{m_{w,\a}}}~.
\end{equation}
\end{corollary}

\begin{proof}
Again, $\gen{U_{m,\a}} \cong \integer_{p^m}$ and includes $\gen{U_{m',\a}}$ for all $m' < m$ but not $\gen{U_{m',\a'}}$ for $\a' \neq \a$.  Thus, each value of $\a$ with $\wt{\a} \leq w$ gives one factor of $\integer_{m_{w,\a}}$.  There is also a $U(1)$ factor from the global phase.
\end{proof}

\section{Conclusion}
\label{sec:conclusion}

We have given a complete characterization of the diagonal elements of the Clifford hierarchy in terms of polynomials and $p^m$-th roots of unity.  One interesting aspect of this result is that it shines light on the distinction between the qubit Clifford group and the qudit Clifford groups.  $\cliff{k}_d$ over qudits of dimension $p$ involves only $p$-th roots of unity for $k < p$.  It is only when $k=p$ do we need other roots of unity.  For qubits, this change is already appearing at $k=2$, the Clifford group, whereas for larger $p$ it is delayed into the more exotic higher levels of the Clifford hierarchy.

\section{Acknowledgements}
We would like to thank Mark Howard for discussions and pointing out an error in an
earlier version of this paper.  This research was supported in part by CIFAR and by
the Perimeter Institute for Theoretical Physics. Research at Perimeter Institute is
supported by the Government of Canada through the Department of Innovation, Science
and Economic Development Canada and by the Province of Ontario through the Ministry
of Research, Innovation and Science.

\bibliographystyle{abbrv}
\bibliography{references}

\end{document}